\documentclass[10pt]{article}
\usepackage[T1]{fontenc}
\usepackage[cp1250]{inputenc}
\usepackage{lmodern}

\usepackage[english]{babel}
\usepackage{latexsym}
\usepackage{amsmath}
\usepackage{indentfirst}
\usepackage{amsthm}
\usepackage{amssymb}
\usepackage{amsfonts}
\usepackage{stackrel}
\usepackage{dsfont}
\usepackage{tikz}
\usepackage{slashbox}
\usepackage[mathscr]{eucal}
\usepackage{authblk}
\usepackage{hyperref}
\usepackage{mathabx}
\usepackage{bookmark}

\newtheorem{theorem}{Theorem}
\newtheorem{lemma}[theorem]{Lemma}

\newcommand{\reals}{\mathbb{R}}
\newcommand{\naturals}{\mathbb{N}}

\begin{document}
\title{Majority rule as a unique voting method in elections with multiple candidates}
\author{Mateusz Krukowski}
\affil{Institute of Mathematics, \L\'od\'z University of Technology, \\ W\'ol\-cza\'n\-ska 215, \
90-924 \ \L\'od\'z, \ Poland \\ \vspace{0.3cm} e-mail: mateusz.krukowski@p.lodz.pl}
\maketitle

\begin{abstract}
May's classical theorem states that in a single-winner choose-one voting system with just two candidates, majority rule is the only social choice function satisfying anonimity, neutrality and positive responsiveness axiom. Anonimity and neutrality are usually regarded as very natural constraints on the social choice function. Positive responsiveness, on the other hand, is sometimes deemed too strong of an axiom, which stimulates further search for less stringent conditions. One viable substitute is Gerhard J. Woeginger's ``reducibility to subsocieties''. We demonstrate that the condition generalizes to more than two candidates and, consequently, characterizes majority rule for elections with multiple candidates. 
\end{abstract}


\section{Introduction}
\label{section:intro}

``No man is an island entire of itself; every man is a piece of the continent, a part of the main''. These two lines from John Donne's ``Devotions Upon Emergent Occasions'' encompass the simple truth that humans are social creatures and rarely do they thrive without societal structures. We cannot outwrestle a bear, outrun a cheetah or outswim a sea lion but the ability to communicate ideas and form complex hierarchies has led us, descendants of bare-foot hunter-gatherers of the African savannah, to complete dominion over all forms of life on the planet. Had we not learnt to organize and work in large groups, we would still be hiding in the tree crowns, trembling in fear not to fall prey to the predators roaming the land.

Not surprisingly, an endeavour to erect stable societal structures and hierarchies is riddled with challenges. Many of them stem from the fact that individual people have conflicting individual preferences and yet, some form of agreement or decision has to be made. History has no shortage of examples when supporters of one opinion convinced people with opposite views using spears, swords, arrows, guns or tanks. Luckily, over the ages people have also worked out less savage and brutal methods of settling debates. 

One of those methods is simply to take a vote and check which of the candidates or available options is the most popular amongst the voters. Such a system is called a majority voting rule and in 1952, Kenneth O. May proved that in the case of two candidates, it is the only voting method which discriminates neither the voters nor the candidates and satisfies a sort of tie-breaking condition (see \cite{May}). In order to build solid grounds for later sections, let us pause and contemplate May's result in detail. We strive to convey the spirit of May's theorem, although the notation and mathematical description is modernized and tailored to our needs in the sequel.   

Suppose we have $n$ voters (labeled $k=1,\ldots,n$) and only two ``candidates''. These can be two political parties, two presidential candidates or two alternatives whether to increase taxes or keep the status quo. Either way, we will refer to these possibilities as ``candidates'', for the sake of simplicity. We represent the two candidates by standard basis vectors $e_1$ and $e_2$ in the Euclidean space $\reals^2.$ If the $k-$th voter supports candidate $1$, we write $P_k = 1,$ if they support candidate $2$, we write $P_k = e_2$. There is a third option in which the voter abstains from voting altogether -- we describe this scenario with $P_k = 0.$ We gather individual votes into a matrix $P = (P_1,\ldots,P_n),$ which we call a voting profile. Now, the core problem of the voting assembly is: given a profile $P$, how do we choose between the two candidates in a ``fair'' manner? In mathematical terms, we want to construct a social choice function 
$$f : \bigcup_{n\in\naturals}\ \{0,e_1,e_2\}^n \longrightarrow \{0,e_1,e_2\}$$

\noindent
which reflects the decision of the entire congregation of $n$ voters and is, in some sense of the word, ``just''. Obviously, justice is not a mathematical term, so May came up with three conditions on the social choice function, which can serve as an approximation to the notion of fairness. Two of them are widely accepted till this day, whereas the third condition is a bit more controversial, as we shall shortly discuss. 

The first of May's conditions is referred to as anonimity and states that all votes should carry equal weight. In other words, the social choice function $f$ should be egalitarian in the sense that it does not matter whom votes for a given candidate. A vote is a vote, irrespective of whether it was cast by Mr Smith or Mr Jones. To describe anonimity in mathematical parlance, we make use of matrices which are formed as a permutation of vectors $e_1,\ldots,e_n.$ For any profile $P$ and permutation matrix $\sigma$, the profile $P\sigma$ is just a permutation of the original votes. Anonimity axiom states that such a permutation should not affect the outcome of the voting, which translates to the social choice function $f$ satisfying $f(P\sigma) = f(P)$ for every profile $P$ and permutation matrix $\sigma$.  

The second condition on May's list pertains to the candidates rather than the voters. It states that if we simultaneously swap the order of the candidates on all ballots, then we should also swap the winner. Mathematically, neutrality means that $f(\sigma P) = \sigma f(P)$ for every profile $P$ and permutation matrix $\sigma.$ Since May considered only two candidates, the permutation matrix can either be an identity matrix or a swap (transposition) of candidates.

In order to explain the third of May's conditions, suppose we have an election with a voting profile $P$, which either results in a draw or crowns candidate 1 as the winner. May's tie-breaking condition states that if any ``dissident'' (who either abstained from voting, or voted for the losing candidate 2) changed his mind, then the new profile $P'$ would crown candidate 1 as the winner. To put it differently, if there is a tie or candidate 1 is already winning, then more votes in his favour cannot make him a loser. Obviously, the condition applies to both candidates in equal measure.

In order to formulate positive responsiveness in mathematical terms, we use the notation $P' >_k P$ if there exists a voter $l^*=1,\ldots,n$ such that 
\begin{itemize}
	\item $P'_l = P_l$ for every $l \neq l^*,$
	\item $P_{l^*} \neq e_k,$
	\item $P'_{l^*} = e_k$ for some $k=1,2.$
\end{itemize}

\noindent
We say that the social choice function $f$ satisfies positive responsiveness, if $f(P) = e_k$ implies $f(P') = e_k$ for every $k=1,2$ and profiles $P, P'$ such that $P' >_k P.$ 

With these three conditions at his disposal, May proved that the only social choice function which satisfies them all simultaneously is the majority rule, i.e., the candidate with higher number of votes wins or there is a tie if two candidates get an equal number of votes. In other words, if we concede that our voting method should not discriminate any voters or candidates and that it should be ``monotone'', then we have to go with majority rule as our social choice function. 

As we have alluded to earlier, both anonimity and neutrality are widely accepted as very natural conditions. It is the third of May's conditions, positive responsiveness, that stirs controversy in the research community. In 2000 Campbell and Kelly went as far as to conclude that ``it is not at all clear why it [positive responsiveness axiom] should be imposed.'' Such attitude prompted both mathematicians and economists alike to look for possible substitutes, oftentimes with success:

\begin{itemize}
	\item In 1983 Fishburn replaced positive responsiveness with its mirror reflection, limited responsiveness (see \cite{Fishburn83}). To quote his own words, limited responsiveness states that ``if the issue passes when the vote is $x$, it will not be defeated, but could be unresolved, if one voter had abstained instead of voting `for', or had voted `against' instead of abstaining.''
	
	\item In 2002 Asan and Sanver replaced positive responsiveness with two ``milder'' conditions: Pareto optimality and weak path independence (see \cite{AsanSanver2002}). 
	
	\item In 2003 Woeginger came up with a substitute for positive responsiveness, which he called ``reducibility to subsocieties'' (see \cite{Woeginger}).
	
	\item In 2004, Miroiu presented a new axiomatization of majority rule that appealed to properties, which he called ``null society'' and ``subset decomposability'' (see \cite{Miroiu}). He did not get rid of monotonicity condition altogether, replacing positive responsiveness with additive responsiveness. 

	\item In 2006, Asan and Sanver explored Maskin monotonic aggregation rules with respect to the family of absolute qualified majority rules (see \cite{AsanSanver2006}).
	
	\item In 2019 Alcantud dismissed May's axioms altogether and suggested his original conditions: individual consistency, non-swap and social fairness (see \cite{Alcantud}).
\end{itemize}

Our paper is motivated by the work of Woeginger, as well as yet another paper by Robert E. Goodin and Christian List (see \cite{GoodinList}). In their publication the authors posed a very important and almost self-provoking question regarding May's result: what happens if we remove the restriction regarding only two candidates? After all, many elections in which we participate have more than just two alternatives (there are over dozen countries which incorporate majority rule as their voting method for the heads of state: Angola, Cameroon, Democratic Republic of Congo, Honduras, Iceland, Mexico, Nicaragua, Palestine, Panama, Paraguay, Rwanda, Singapore, South Korea, or Venezuela). 

Our goal is to elevate the contribution of Woeginer to the case of multiple-candidate elections. In other words, we want to investigate elections with more than two candidates, but at the same time, we follow the trend of avoiding such stringent conditions as positive responsiveness. To our rescue comes Woeginger's reducibility to subsocieties, which we generalize to encompass the case of multiple candidates. We devote the next section to formulating a suitable mathematical foundations for our endeavour.

\section{Beyond two candidates: a mathematical toolbox}	
\label{section:toolbox}
	
We consider $n$ voters (numbered with $l = 1,\ldots,n$) and $m\geqslant 2$ candidates (again, these can be either individual people, political parties, tax alternatives etc.). We associate the candidates with the standard basis in $\reals^m: e_1,\ldots,e_m.$ This correspondence allows us to write $P_1 = e_3$ if the first voter supports candidate $3$, $P_2 = e_5$ if the second voter supports candidate $5$ etc. Additionally, each voter can abstain from voting altogether, in which case we write $P_l = 0 \in\reals^m.$  

All the individual votes (i.e., vectors in $\reals^m$) $P_l,\ l=1,\ldots,n$ are compiled into matrix $P = (P_1|P_2|\ldots|P_n)$ called a voting profile. A social choice function is a map
$$f: \bigcup_{n\in\naturals}\ \{0,e_1,\ldots,e_m\}^n \longrightarrow \{0,e_1,\ldots,e_m\},$$ 

\noindent
which assigns a single outcome to every possible profile. In the introductory section we have already discussed two conditions that we will impose on the social choice function:

\begin{description}
	\item[\hspace{0.4cm} Anonimity (A).] We say that a social choice function $f$ is anonymous, if it does not favour or discriminate any of the voters, i.e., for any voting profile $P$ and any permutation matrix $\sigma$ we have $f(P\sigma) = f(P).$
 	
	\item[\hspace{0.4cm} Neutrality (N).] We say that a social choice function $f$ is neutral, if it does not favour or discriminate any of the candidates, i.e., for any voting profile $P$ and any permutation matrix $\tau$ we have $f(\tau P) = \tau f(P).$ 
\end{description}

Our next two axioms do not appear in the original May's theorem, because May used the much stronger condition of positive responsiveness. However, we feel that the two presented properties are so ``natural'' and ``justified'' that they should not stir much controversy. 

\begin{description}
	\item[\hspace{0.4cm} Duel property (DP).] A voting profile $P$, which consists solely of columns $e_i,e_j$ or $0$, is called a duel profile for candidates $i$ and $j$. We say that a social choice function $f$ satisfies the duel property, if the duel can only result in a tie or a win by one of the duellists. No third party can emerge victorious from a duel they did not participate in. To formalize the duel property mathematically, let $\Sigma_k(P)$ stand for the number of vectors $e_k$ amongst the columns of the voting profile $P.$ In other words, $\Sigma_k(P)$ measures the number of votes in $P$ amassed by the $k-$th candidate. The outcast property states that if $P$ is a duel profile between some candidates $i,j=1,\ldots,m$ then $f(P) = e_i,e_j$ or $0$.
	
	\item[\hspace{0.4cm} Pareto optimality (PO).] A voting profile $P$ is called a Pareto profile for candidate $k^*=0,1,\ldots,m$ if all voters either vote for this candidate or they abstain from voting altogether, i.e. $\Sigma_{k^*}(P) > 0$ and $\Sigma_k(P) = 0$ for every $k\neq 0,k^*.$ We say that a social choice function $f$ is Pareto optimal if whenever $P$ is a Pareto profile for some candidate $k^*$ then $f(P) = e_{k^*}.$
\end{description}

It is worth emphasizing that the duel property seems to be essential only if there are just $m=3$ candidates. If there are more than $3$ competitors, then \textbf{(DP)} follows from \textbf{(N)}:

\begin{lemma}
If there are at least $4$ candidates (i.e., $m\geqslant 4$), then any social choice function satisfying \textbf{(N)} also satisfies \textbf{(DP)}.
\label{DPredundant} 
\end{lemma}
\begin{proof}
Suppose, for the sake of contradiction, that $f(P) = e_k,$ where $P$ is a duel profile for candidates $i$ and $j$. Let $\tau$ be a permutation matrix which swaps candidate $k$ with $l\neq i,j$, leaving the remaining competitors (in particular $i$ and $j$) in their original places. Then 
$$f(\tau P) = f(P) = e_k$$

\noindent
but at the same time
$$\tau f(P) = \tau e_k = e_l.$$

\noindent
By \textbf{(N)} we get a contradiction $e_k = e_l,$ so the social choice function $f$ has to satisfy \textbf{(DP)}.
\end{proof}

It is clear that what makes the proof work is ample room in the space of candidates to swap candidate $k$ with ``some other'' candidate $l$, which is different from $i$ and $j$. This line of reasoning cannot be applied to $m = 3$ candidates.

We have now compiled the mathematical toolbox for future work. However, the introduced concepts are not sufficient to characterize the majority voting rule. For instance, consider a unanimity consent voting rule defined by 
\begin{gather}
UC(P) := \left\{\begin{array}{cl}
e_k & \text{ if $P$ is a Pareto profile for candidate $k$,}\\
0 & \text{ otherwise.}
\end{array}\right.
\label{unanimousconsent}
\end{gather}

\noindent
This means that $k-$th candidate becomes the winner if and only if no one voted for any other candidate (although abstaining from voting is permissible). In any other scenario, i.e., if there are at least two candidates in the voting profile $P$, the unanimity consent voting rule declares a tie. It is easy to see that this function satisfies all the above axioms, i.e., anonimity, neutrality, the duel property and Pareto optimality. However, it is also quite clear that the unanimity consent is different from the majority rule. This means that there are still some ``missing puzzles'' that require investigation, before we will be able to lay down a complete description of the majority voting rule. This is why section \ref{section:reducibility} focuses on Woeginger's condition of ``reducibility to subsocieties'', which constitutes a viable way to fill in the gaps and characterize the majority rule.

\section{Reducibility to subsocieties}
\label{section:reducibility}

Contemplating the classical May's theorem, one may arrive at a conclusion that allowing for just two candidates is a major restriction. After all, many elections we take part in involve multiple candidates. Goodin and List were well-aware of this issue (see \cite{GoodinList}) and managed to prove that May's theorem still holds for multiple candidates provided we extend the notions of anonimity, neutrality and positive responsiveness in a suitable manner. We have already described the fitting generalizations of anonimity and neutrality in Section \ref{section:toolbox}. As for the positive responsiveness, the Reader will surely guess its ``correct'' extension given our exposition in Section \ref{section:intro}. 

Irrespective of Goodin and List's generalization of May's theorem, some economists and mathematicians still believe that positive responsiveness is too restrictive and stringent. Consequently, researchers strive to come up with more flexible and relaxed conditions. In this section we focus on Woeginger's axiom called ``reducibility to subsocieties''. We take the liberty of modifying his original condition in order to take into account more than two candidates. 

\begin{description}
	\item[\hspace{0.4cm} Reducibility to subsocieties (RS).] We say that a social choice function $f$ satisfies reducibility to subsocieties if for every voting profile $P$ with at least $2$ voters we have 
		$$f(P) = f(f(P^{-1}) | \ldots | f(P^{-n})),$$
		
		\noindent
		where $P^{-l}$ denotes profile $P$ with $l-$th column removed. Woeginger argued that ``in certain university councils (...), when the chairman is monitoring a ballot, then sometimes he abstains from voting himself in order to keep the procedure transparent. Now assume that over time, the chairmanship circulates through the council, and that the voting rule should express the average opinion of all such councils under all the possible chairmen. Then the council without chairman forms a subsociety of the council, (...) and RS stipulates that the aggregate preference of the total council behaves in the same way as the aggregate preference of all the subsocieties under changing chairmen.''	
\end{description}

As we are about to demonstrate, \textbf{(RS)} is the ``missing puzzle'' allowing for the characterization of the majority rule. However, prior to proving the main result of the section, we need the following auxiliary result:

\begin{lemma}
Let $f$ be a social choice function satisfying \textbf{(A)} and \textbf{(N)}. If the voting profile $P$ results in a tie between (at least) two candidates $i$ and $j$, then neither of them can be a winner, i.e., if there exist distinct $i,j =1,\ldots,m$ such that $\Sigma_i(P) = \Sigma_j(P)$ then $f(P)\neq e_i, e_j.$ \label{TiersCannotBeWinners}
\end{lemma}
\begin{proof}
Let $P$ be a tie between (at least) two candidates $i$ and $j$. We suppose, for the sake of contradiction, that $f(P) = e_i$ and let $\tau$ be a matrix which swaps rows $i$ and $j$, leaving the remaining rows unaltered. Since $P$ is a tie we have $\tau P = P \sigma$ for some matrix $\sigma$ corresponding to interchanging columns $e_i$ with $e_j$ in the profile $P$. We have 
$$f(\tau P) = f(P \sigma) \stackrel{\textbf{(A)}}{=} f(P) = e_i,$$

\noindent
but on the other hand 
$$\tau f(P) = \tau e_i = e_j.$$

\noindent
This implies, due to \textbf{(N)}, that $e_i = e_j,$ which is the desired contradiction. The possibility $f(P) = e_j$ is dismissed analogously. 
\end{proof}

\begin{theorem}
A social choice function satisfies \textbf{(A)}, \textbf{(N)}, \textbf{(DP)}, \textbf{(PO)} and \textbf{(RS)} if and only if it is a majority voting rule
\begin{gather*}
Maj(P) := \left\{\begin{array}{cl}
e_{k^*} & \text{ if }\ \Sigma_{k^*}(P) > \Sigma_k(P)\ \text{ for every }\ k\neq 0,k^*,\\
0 & \text{ otherwise.}
\end{array}\right.
\end{gather*}

\noindent
Furthermore, \textbf{(DP)} is redundant for $m\geqslant 4$ candidates.
\label{WoegingerMay}
\end{theorem}
\begin{proof}
One direction of the statement is straightforward, since the majority rule clearly satisfies all axioms \textbf{(A)}, \textbf{(N)}, \textbf{(OA)}, \textbf{(PO)} and \textbf{(RS)}. In order to prove the reverse implication we use the mathematical induction on the number of voters. If there is just one voter, i.e., $n = 1$, then $f(e_k) = e_k$ for every $k=1,\ldots,m$ by \textbf{(PO)}. Furthermore, by \textbf{(N)} we have 
$$\tau f(0) = f(\tau 0) = f(0)$$

\noindent
for every permutation matrix $\tau$. This implies that $f(0) = 0$ and, consequently, $f = Maj$ for $n = 1.$

Next, we assume that $f = Maj$ for $n-1$ voters and demonstrate that the equality holds for $n$ voters as well. We divide our reasoning into two separate cases. First, it can happen that two (or more) candidates in the profile $P$ get an equal number of votes, which is greater than any other number of votes (we call it a dominating tie). Mathematically: there exist two distinct candidates $i,j = 1,\ldots,m$ such that for every $k\neq 0,i,j$ we have $\Sigma_i(P) = \Sigma_j(P) > \Sigma_k(P).$  

The second scenario occurs when there is a clear leader amongst the candidates in the profile $P$. This leader accumulates a higher number of votes than any other candidate. Mathematically: there exists a candidate $k^*=1,\ldots,m$ such that $\Sigma_{k^*}(P) > \Sigma_k(P)$ for every $k\neq 0,k^*.$\\
\vspace{0.3cm}

\noindent
\textbf{Case 1: Dominating tie.}

By the inductive assumption we have 
\begin{gather*}
f(P^{-k}) = Maj(P^{-k}) = \left\{
\begin{array}{cl}
e_j & \text{if }\ P_k = e_i\ \text{ and }\ \Sigma_j(P^{-k}) > \Sigma_l(P^{-k})\ \text{ for every }\ l\neq 0,j,\\
e_i & \text{if }\ P_k = e_j\ \text{ and }\ \Sigma_i(P^{-k}) > \Sigma_l(P^{-k})\ \text{ for every }\ l\neq 0,i,\\
0 & \text{otherwise.}
\end{array}\right.
\end{gather*}

\noindent
Invoking \textbf{(RS)}, the above implies that 
$$f(P) = f(f(P^{-1}), \ldots, f(P^{-n})) = f(P')$$

\noindent
where $P'$ is a voting profile consisting of columns $e_i, e_j$ and $0$. Due to \textbf{(DP)} $f(P')$ can only be $e_i, e_j$ or $0$, since $\Sigma_k(P') = 0$ for every $k\neq 0,i,j.$ On the other hand, $f(P)$ can be neither $e_i$ nor $e_j$ due to Lemma \ref{TiersCannotBeWinners}. Hence, it follows that $f(P) = 0 = Maj(P).$\\
\vspace{0.3cm}

\noindent
\textbf{Case 2: Leader.}

By the inductive assumption we have 
\begin{gather*}
f(P^{-k}) = Maj(P^{-k}) = \left\{
\begin{array}{cl}
e_{k^*} & \text{if }\ P_k = e_{k^*}\ \text{ and }\ \Sigma_{k^*}(P^{-k}) > \Sigma_l(P^{-k})\ \text{ for every }\ l\neq 0,k^*,\\
0 & \text{otherwise.}
\end{array}\right.
\end{gather*}

\noindent
Invoking \textbf{(RS)}, the above implies that 
$$f(P) = f(f(P^{-1}), \ldots, f(P^{-n})) = f(P')$$

\noindent
where $P'$ is a voting profile consisting of columns $e_{k^*}$ and $0$. Due to \textbf{(PO)} we have $f(P') = e_{k^*}$, which concludes the first part of the proof. It remains to note that the redundancy of \textbf{(DP)} for $m\geqslant 4$ follows from Lemma \ref{DPredundant}.  
\end{proof}

Trying to tie up loose ends, let us investigate the relations between some of the axioms. We have already demonstrated that the duel property is a consequence of neutrality if $m \geqslant 4$. Furthermore, we have the following result:

\begin{theorem}
The axioms \textbf{(N)}, \textbf{(PO)} and \textbf{(RS)} are independent, i.e., there exist (anonymous) social choice functions $f_1,f_2,f_3$ such that 
\begin{enumerate}
	\item $f_1$ satisfies \textbf{(A)}, \textbf{(DP)}, \textbf{(PO)} and \textbf{(RS)} but not \textbf{(N)},
	\item $f_2$ satisfies \textbf{(A)}, \textbf{(N)}, \textbf{(DP)} and \textbf{(RS)} but not \textbf{(PO)},
	\item $f_3$ satisfies \textbf{(A)}, \textbf{(N)}, \textbf{(DP)} and \textbf{(PO)} but not \textbf{(RS)}.
\end{enumerate}
\end{theorem}
\begin{proof}
\begin{enumerate}
	\item Let 
	\begin{gather*}
	f_1(P) := \left\{\begin{array}{cl}
	0 & \text{ if }\ P = (0|\ldots|0),\\
	e_{k^*} & \text{ if }\ \Sigma_{k^*}(P) > 0\ \text{ and }\ \Sigma_k(P) = 0\ \text{ for every }\ 0<k<k^*.\\
	\end{array}\right.
	\end{gather*}
	
	\noindent
	Intuitively, $f_1$ introduces an ``artificial'' order between the candidates (thus breaking neutrality axiom) and declares the first person with nonzero number of votes as the winner. It is easy to see that $f_1$ satisfies \textbf{(A)}, \textbf{(DP)} and \textbf{(PO)}. \textbf{(RS)} is satisfied if $P$ is a zero profile, so let $P$ be such that $\Sigma_{k^*}(P) > 0$ and $\Sigma_k(P) = 0$ for every $0<k<k^*.$ This implies $f_1(P) = e_{k^*}$ and by \textbf{(A)} we may, without loss of generality, assume that $e_{k^*}$ is supported by the first voter, i.e., $P_1 = e_{k^*}.$ Next, we observe that 
	\begin{gather*}
	f_1(P^{-1}) = \left\{\begin{array}{cl}
	e_{k^*} & \text{ if someone besides voter 1 supports } e_{k^*},\\
	e_k & \text{ for some $k > k^*$ if no one besides voter 1 supports } e_{k^*}.
	\end{array}\right.
	\end{gather*}
	
	On the other hand we have $f_1(P^{-l}) = e_{k^*}$ for every $l=2,\ldots,n.$ Hence, we have 
	$$f_1(f_1(P^{-1}) | f_1(P^{-2}) | \ldots | f_1(P^{-n})) = f_1(e_{k,\text{ where } k\geqslant k^*}| e_{k^*} | \ldots | e_{k^*}) = e_{k^*} = f_1(P),$$
	
	\noindent
	which proves that $f_1$ satisfies \textbf{(RS)}.
	
	Finally, in order to prove that $f$ does not satisfy \textbf{(N)} it suffices to check that if $\tau$ is a transposition of first two rows, then 
	$$f_1(\tau (e_1|e_2)) = f_1(e_2|e_1) = e_1,$$
	
	\noindent
	which is different than
	$$\tau f_1(e_1|e_2) = \tau e_1 = e_2.$$
	
	\item It suffices to note that $f_2 = 0$ satisfies \textbf{(A)}, \textbf{(N)}, \textbf{(DP)} and \textbf{(RS)} but not \textbf{(PO)}.
	\item Let $f_3$ be the unanimous consent method (i.e., $f_3 = UC$) defined by \eqref{unanimousconsent} in Section \ref{section:toolbox}. As we have already noted, $f_3$ satisfies \textbf{(A)}, \textbf{(N)}, \textbf{(DP)} and \textbf{(PO)}. In order to see that $f_3$ does not satisfy \textbf{(RS)} take $P = (e_1|e_1|e_2),$ so that $f_3(P) = 0.$ However, we have 
	\begin{equation*}
	\begin{split}
	f_3(f_3(P^{-1}) | f_3(P^{-2}) | f_3(P^{-3})) &= f_3(f_3(e_1|e_2) | f_3(e_1|e_2) | f_3(e_1|e_1)) \\
	&= f_3(0|0|e_1) = e_1,
	\end{split}
	\end{equation*}
	
	\noindent
	which demonstrates that \textbf{(RS)} does not hold. 
\end{enumerate}
\end{proof}

\end{document}